\documentclass{eptcs}
\providecommand{\event}{TARK 2015} 
\usepackage{breakurl}             

\usepackage{color,amssymb,amsmath,amsfonts,wasysym,mathtools}
\usepackage{hyperref}
\hypersetup{colorlinks, linkcolor=blue}
\usepackage{array, multirow}
\DeclareSymbolFont{symbolsC}{U}{txsyc}{m}{n}
\DeclareMathSymbol{\boxright}{\mathrel}{symbolsC}{128}

\usepackage{mylogicv01}
\standardlogic
\def\titlerunning{Ceteris paribus logic in counterfactual reasoning}
\def\authorrunning{P.\ Girard and M.\ A.\ Triplett}

\DeclareMathSymbol{\diamondright}{\mathrel}{symbolsC}{132}

\usepackage{booktabs}

\usepackage{verbatim}
\usepackage{stmaryrd}
\usepackage{mathtools}
\usepackage[noadjust]{cite}
\usepackage{bussproofs}
\newcommand{\counterfactual}{\ensuremath{%
  \Box\kern-1.5pt
  \raise1pt\hbox{$\mathord{\rightarrow}$}}}
\def\mathname#1{\ensuremath{\mathsf{#1}}}
\renewcommand\phi\varphi
\newcommand\lb\llbracket
\newcommand\rb\rrbracket
\def\prop{\mathname{Prop}}
\def\min{\mathname{Min}}

\def\AT{\mathname{CP}}
\def\CP{\mathname{CP}}
\def\NC{\mathname{NC}}
\def\ud{\mathname{UD}}
\def\MS{\mathname{MS}}

\def\MP{\mathname{\MP}}
\usepackage{tikz}
\usetikzlibrary{arrows,decorations.pathmorphing,backgrounds,positioning,fit,automata,
  matrix}
\newenvironment{configuration}{
\begin{tikzpicture}[>=stealth',
state/.style={rectangle, rounded corners=2mm, draw=black!75, text
  centered, inner sep=5pt},
state-empty/.style={circle, fill=black, inner sep=3pt}]}
{\end{tikzpicture}}
\tikzstyle{mybox} = [draw=black,  thick, rectangle, rounded corners, inner sep=5pt,]

\newtheorem{theorem}{Theorem}	
\newtheorem{lemma}{Lemma}

\newtheorem{corollary}{Corollary}
\newtheorem{definition}{Definition}
\newtheorem{fact}{Fact}

\newenvironment{proof}{\begin{trivlist}\item[]{\sc Proof:}\rm}{\hfill $\square$\end{trivlist}} 

\newcommand{\M}{\mathfrak{M}}

\newcommand{\F}{\mathcal{F}}

\title{Ceteris paribus logic in counterfactual reasoning}
\author{Patrick Girard
\institute{University of Auckland}
\email{p.girard@auckland.ac.nz}
\and
Marcus Anthony Triplett
\institute{University of Auckland}
\email{mtri285@aucklanduni.ac.nz}
}
\def\titlerunning{Ceteris paribus logic in counterfactual reasoning}
\def\authorrunning{P. Girard and M. A. Triplett}
\begin{document}
\maketitle

\begin{abstract}
The semantics for counterfactuals due to David Lewis has been
challenged on the basis of unlikely, or impossible, events. Such
events may skew a given similarity order in favour of those possible
worlds which exhibit them. By updating the relational
structure of a model according to a \emph{ceteris paribus} clause one
forces out, in a natural manner, those possible worlds which do not
satisfy the requirements of the clause. We develop a ceteris paribus
logic for counterfactual reasoning capable of performing such actions, and offer several alternative (relaxed) interpretations of \emph{ceteris paribus}. We apply this framework in a way which allows us to reason counterfactually without having our similarity order skewed by unlikely events. This continues the investigation of formal ceteris paribus reasoning, which has previously been applied to preferences \cite{RefWorks:2}, logics of game forms \cite{grossi2013ceteris}, and questions in decision-making \cite{xiong2011open}, among other areas \cite{proietti2010fitch}.
\end{abstract}

\section{Introduction}

The principal task of this paper is to work towards integrating \emph{ceteris paribus} modalities into conditional logics so that some dissonant analyses of counterfactuals may be reconciled. We also suggest that ceteris paribus clauses may be understood dynamically, in the sense of dynamic epistemic logic \cite{van2007dynamic}, and we interpret our resulting ceteris paribus logic accordingly. Ceteris paribus clauses implicitly qualify many conditional statements that formulate laws of science and economics. A ceteris paribus clause adds to a statement a proviso requiring that other variables or states of affairs not explicitly mentioned in the statement are kept constant, thus ruling out benign defeaters. For instance, Avogadro's law says that if the volume of some ideal gas increases then, everything else held equal, the number of moles of that gas increases proportionally. Varying the temperature or pressure could provide situations that violate the plain statement of the law, but the ceteris paribus clause accounts for those. It specifically isolates the interaction between volume and number of moles by keeping everything else equal. In the same spirit, the Nash equilibrium in game theory is a solution concept that picks strategy profiles in which none of the agents could unilaterally (i.e., keeping the actions of others constant, or equal) deviate to their own advantage.

We may understand a ceteris paribus clause as a linguistic device
intended to shrink the scope of the sentence qualified by the clause. For instance, when I make the utterance ``I prefer fish to beef, \emph{ceteris paribus}" I may mean something different from if I simply uttered ``I prefer fish to beef." By enforcing the ceteris paribus condition I rule out some situations which affect my preference. For example if, whenever I eat fish I'm beaten with a mallet, while whenever I eat beef I'm left in peace, I might retract the second utterance and maintain the first. The ceteris paribus clause reduces the number of states of affairs under consideration. For modal logicians, `ruling out' states of affairs amounts to strengthening an accessibility relation, consequently changing the relational structure of a model. This bears similarity to the epistemological forcing of Vincent Hendricks \cite{hendricks2006mainstream}, which seeks to rule out `irrelevant alternatives' in a way which allows knowledge in spite of the possibility of error. Wesley Holliday \cite{holliday2014epistemic} develops several interpretations of the epistemic operator $K$ based on the relevant alternatives epistemology; namely, that in order for an agent to have knowledge of a proposition, that agent must eliminate each \emph{relevant alternative}. Holliday's semantics are based on the semantics for counterfactuals due to David Lewis \cite{RefWorks:170}. One could see relevant worlds as those which keep things equal. When reasoning using Avogadro's law, the relevant possible worlds are those where the temperature and pressure have not changed. Thus, in order for an agent to have knowledge, that agent must eliminate the alternatives among the worlds which `keep things equal.'

Previously, ceteris paribus formalisms have been given for logics of preference \cite{RefWorks:2} and logics of game forms \cite{grossi2013ceteris}. Here we extend the analysis to counterfactual reasoning. The importance of counterfactuals in game theory is well known (see, for instance, \cite{samet1996hypothetical}). For example, Bassel Tarbush \cite{tarbush2013agreeing} argues that the \emph{Sure-Thing Principle}\footnote{An outcome $o$ of an action $A$ is a \emph{sure-thing} if, were any other action $A'$ to be chosen, $o$ would remain an outcome. The Sure-Thing Principle \cite{savage1972foundations} states that sure-things should not affect an agent's preferences.} ought to be understood as an inherently counterfactual notion. We will motivate our discussion by thinking through  Kit Fine's well-known `minor-miracles' argument \cite{fine1975review}, a putative counterexample to Lewis' semantics. We will argue that ceteris paribus logic, suitably adapted to conditionals, provides a natural response to this kind of argument. Moreover, we will see that ceteris paribus logic reveals a useful feature missing from the standard formalisation of counterfactuals; namely, the explicit requirement that certain propositions must have their truth remain fixed during the evaluation of the counterfactual. This is implicitly thought to hold, to some degree, when one works with models which have similarity orders or systems of spheres. The conditional logic of Graham Priest \cite{priest2008introduction} makes just that assumption, but with no syntactic assurance. Ceteris paribus logic provides, in addition to the underlying similarity order over possible worlds, a syntactic apparatus to reason with such ceteris paribus clauses directly in the object language. 

\section{Counterfactuals}

Here we shall formalise counterfactuals in the style of Lewis. Let $\prop$ be a set of propositional variables. 
We are concerned with models of the form $\M = (W, \preceq, V)$ such that the following obtain. \begin{enumerate}
\item $W$ is a non-empty set of \emph{possible worlds}.
\item $\preceq$ is a family $\{\preceq_w\}_{w\in W}$ of
  \emph{similarity orders}, i.e., relations on $W_w \times W_w$ (with $W_w \subseteq W$)  such that: 
\begin{itemize}
 \item $w \in W_w$, 
 \item $\preceq_w$ is reflexive,  transitive and total, and
 \item $w \prec_w v$ for all $v \in W_w\setminus \{w\}$.
\end{itemize} 
  
\item $V$ is a \emph{valuation function} assigning a subset $V(p)
  \subseteq W$ to each propositional variable $p\in \prop$.
\end{enumerate}
Intuitively, $W_w$ is the set of worlds which are entertainable from
$w$. Worlds which are not entertainable from $w$ are deemed simply too dissimilar from
$w$ to be considered.  Say that $u$ is at least as similar to $w$ as
$v$ is when $u \preceq_w v$, and that it is strictly more similar when
$u\prec_w v$.

If $\M$ satisfies each of the three requirements we call $\M$ a \emph{conditional model}. A relation $\le$ is said to be \emph{well-founded} if for every non-empty $S \subseteq W$ the set

\begin{align} \label{defmin}
\min^\M_\le(S) = \{v \in S \cap W : \text{ there is no $u$ with } u < v\}
\end{align}

\noindent is non-empty.\footnote{As usual, $u < v$ is defined as $u \le v$ and
  not $v \le u$} We will suppress the superscript $\M$ if it is clear
from the context which model we're discussing.  If a model $\M = (W, \preceq, V)$ has only
well-founded similarity orders we say that $\M$ satisfies the
\emph{limit assumption}. For ease of exposition, we will assume that our conditional models satisfy the limit assumption. Of course, we may generalise the semantics for counterfactuals in the usual way \cite{RefWorks:170}, so that our results work for models which do not satisfy the limit assumption as well.

\begin{definition}[Language $\mathcal{L}^\boxright$] The language $\mathcal{L}^\boxright$ of counterfactuals is given by the following grammar
\begin{align*} \phi \ ::= \ p \ | \ \lnot \phi \ | \ \phi \lor \psi \ | \ \phi \boxright \psi. \end{align*}
\end{definition}
We define $\phi \land \psi := \lnot(\lnot \phi \lor \lnot \psi)$, $\phi \rightarrow \psi := \lnot \phi \lor \psi$, $\phi \diamondright \psi := \lnot(\phi \boxright \lnot \psi)$.

\begin{definition}[Semantics] \label{defcfwf} Let $\M = (W, \preceq, V)$ be a well-founded conditional model. Then 

\begin{center}
\begin{tabular}{r c l}
$\lb p \rb^\M$ & = & $V(p)$ \\
$\lb \lnot \phi \rb^\M$ & = & $W \setminus \lb \phi \rb^\M$ \\
$\lb \phi \lor \psi \rb^\M$ & = & $\lb \phi \rb^\M \cup \lb \psi \rb^\M$ \\
$\lb \phi \boxright \psi \rb^\M$ & = & $\{w \in W : \min_{\preceq_w}(\lb \phi \rb^\M) \subseteq \lb \psi \rb^\M\}$.
\end{tabular}
\end{center}
\end{definition}
Let $w \in W$. If $w \in \lb \phi \rb$ we write $\M, w \models \phi$, and if $w \not \in \lb \phi \rb$ we write $\M, w \not \models \phi$.

\section{The Nixon argument}

There is a problem dating back to the 1970s \cite{fine1975review,
  bennett1974review, bowie1979similarity} surrounding the semantics
for counterfactuals proposed by Lewis. We have found that our `ceteris
paribus counterfactuals' (defined below) provide a unique perspective
on the problem (a putative counterexample). The argument goes as
follows. Assume, during the Cold War, that President Richard Nixon had
access to a device which launches a nuclear missile at the
Soviets. All Nixon is required to do is press a button on the
device. Consider the counterfactual \emph{if Nixon had pushed the
  button, there would have been a nuclear holocaust}. Call it  the \emph{Nixon couterfactual}. It is not so
difficult to see that the Nixon counterfactual could be true, or could be
imagined to be true. Indeed, one could argue that the Nixon counterfactual ought to be true in any successful theory of counterfactuals. Fine and Lewis both agree (and so do we) that the counterfactual is true
(\cite[p.~452]{fine1975review},
\cite[p.~468]{lewis1979counterfactual}), but Fine used the Nixon counterfactual to argue that
the Lewis semantics yields the wrong verdict. This is because ``a world with a single miracle but no
holocaust is closer to reality than one with a holocaust but no
miracle.'' \cite[p.~452]{fine1975review} In response, Lewis argues that, provided the Nixon situation is modelled using a similarity relation which respects a plausible system of priorities (see below), the counterfactual will emerge true. We will provide a different response using ceteris paribus counterfactuals, but first let us see how Fine and Lewis model the situation. 

Consider two classes of possible worlds. One class,
$\textbf{u}$, consists of those worlds in which Nixon pushes the
button, and the button successfully launches the missile. The second,
$\textbf{v}$, consists of those worlds in which Nixon pushes the
button, but some small occurrence -- such as a minor miracle --
prevents the button's correct operation. Certainly those worlds where
the button does \emph{not} launch the missile bear more similarity to
the present world than those where it does. This is Fine's
interpretation of Lewis' semantics. Any world in $\textbf{u}$ has
been devastated by nuclear warfare, countless lives have been lost,
there is nuclear winter, etc., whereas worlds in $\textbf{v}$
continue on as they would have done. 

To illustrate Fine's interpretation, let $p, s, m, h$ be the propositions:

\medskip 

\begin{center}
\begin{tabular}{l c l}
$p$ & = & ``Nixon \underline{p}ushes the button," \\
$s$ & = & ``the missile \underline{s}uccessfully launches," \\
$m$ & = & ``a \underline{m}iracle prevents the missile being launched," \\
$h$ & = & ``a nuclear \underline{h}olocaust occurs,"
\end{tabular}
\end{center}

\medskip

\noindent and consider the following model, the \emph{Fine model}:

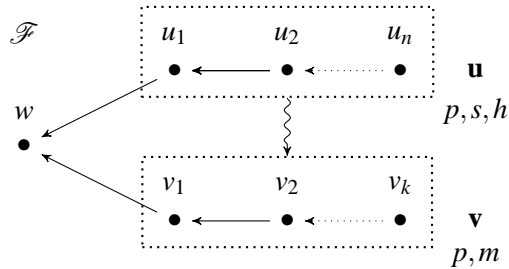
\begin{center}
\begin{tabular}{c}
\begin{configuration}
\node (1) at (0,0) {\textbullet};
\node [above] at (1.north) {$w$};

\node (M) at (0,1.5) {$\mathcal{F}$};

\node (2) at (2,1) {\textbullet};
\node [above] (2t1) at (2.north) {$u_1$};
\node (2a) at (3.5, 1) {\textbullet};
\draw [->] (2a) to (2);
\node [above] (2t2a) at (2a.north) {$u_2$};
\node (2b) at (5, 1) {\textbullet};
\draw [->, dotted] (2b) to (2a);
\node [above] (2t2b) at (2b.north) {$u_n$};
\node (2c) at (6, 1) {$\textbf{u}$};
\node [below] (2p1) at (2c.south) {$p, s, h$};

\node (3) at (2, -1) {\textbullet};
\node [above] (3t1) at (3.north) {$v_1$};
\node (3a) at (3.5, -1) {\textbullet};
\draw [->] (3a) to (3);
\node [above] (3t3a) at (3a.north) {$v_2$};
\node (3b) at (5, -1) {\textbullet};
\draw [->, dotted] (3b) to (3a);
\node [above] (3t3b) at (3b.north) {$v_k$};
\node (3c) at (6, -1) {$\textbf{v}$};
\node [below] (3p1) at (3c.south) {$p, m$};

\draw [->] (2a) to (2);
\draw [->] (2a) to (2);
\draw [->] (2) to (1);
\draw [->] (3) to (1);
\node [draw, dotted, thick, fit=(2) (2a) (2b) (2t1) (2t2a) (2t2b) ] (u) {};
\node [draw ,dotted, thick, fit=(3) (3a) (3b) (3t1) (3t3a) (3t3b)] (v) {};
\draw [->,decorate,decoration={snake,amplitude=.4mm,segment
  length=2mm, post length=1mm}] (u) to (v);
\end{configuration}
\end{tabular}
\end{center}

An arrow from $x$ to $y$ indicates relative similarity to $w$, so
$u_1$ is more similar to $w$ than $u_2$ is. Arrows are transitive, and
the `snake' arrow between $\textbf{v}$ indicates
that $v_i \preceq_w u_j$ for every $i,j$.  For each $u_i \in \textbf{u}$, $\mathcal{F}, u_i \models p \land s \land h$; and for each $v_i \in \textbf{v}$, $\mathcal{F}, v_i \models p \land m$. World $w$ is intended to represent the real world: Nixon did not push any catastrophic anti-Soviet buttons,\footnote{Although there is no way for us to know this, for the sake of the argument we assume that it is so.} no nuclear missile was successfully launched at the Soviets, no miracle prevented any such missile, and no nuclear holocaust occurred.   World $v_1$ is more similar to $w$ than any world in $\textbf{u}$ is, since in any $\textbf{u}$-world Nixon pushes the button and begins a nuclear holocaust. By (\ref{defmin}), $v_1$ is therefore the minimal $p$-world. At $v_1$ the proposition $h$ is false, and so $\mathcal{F}, w \not \models p \boxright h$. Therefore, Fine concludes, the Nixon counterfactual is false in Lewis' semantics. 

In response, Lewis argues that the proper similarity relation to model the Nixon counterfactual should respect the following system of priorities: 

\begin{enumerate}
\item It is of the first importance to avoid big, widespread, diverse violations of law.
\item It is of the second importance to maximize the spatio-temporal region throughout which perfect match of particular fact prevails.
\item It is of the third importance to avoid even small, localized, simple violations of law.
\item It is of little or no importance to secure approximate similarity of particular fact, even in matters that concern us greatly. (\cite[p.~472]{lewis1979counterfactual})
\end{enumerate}

Based on this system of priorities world $u_1$ is more similar to $w$ than $v_1$ is
because ``perfect match of particular fact counts for much more than
imperfect match, even if the imperfect match is good enough to give us
similarity in respects that matter very much to us.''
\cite[p.~470]{lewis1979counterfactual} That is, worlds in $\textbf{v}$
in which a small miracle prevents the missile being launched may look
quite similar to our world, but only approximately so. And in Lewis'
system of priorities, perfect match outweighs approximate
similarity. The \emph{Lewis model}, then, looks like this:

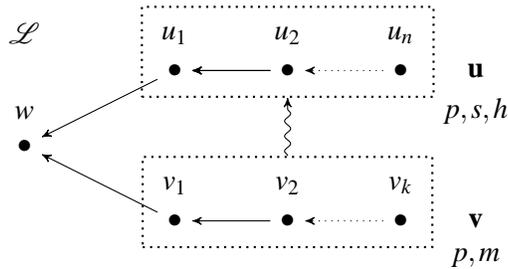
\begin{center}
\begin{tabular}{c}
\begin{configuration}
\node (1) at (0,0) {\textbullet};
\node [above] at (1.north) {$w$};

\node (M) at (0,1.5) {$\mathcal{L}$};

\node (2) at (2,1) {\textbullet};
\node [above] (2t1) at (2.north) {$u_1$};
\node (2a) at (3.5, 1) {\textbullet};
\draw [->] (2a) to (2);
\node [above] (2t2a) at (2a.north) {$u_2$};
\node (2b) at (5, 1) {\textbullet};
\draw [->, dotted] (2b) to (2a);
\node [above] (2t2b) at (2b.north) {$u_n$};
\node (2c) at (6, 1) {$\textbf{u}$};
\node [below] (2p1) at (2c.south) {$p, s, h$};

\node (3) at (2, -1) {\textbullet};
\node [above] (3t1) at (3.north) {$v_1$};
\node (3a) at (3.5, -1) {\textbullet};
\draw [->] (3a) to (3);
\node [above] (3t3a) at (3a.north) {$v_2$};
\node (3b) at (5, -1) {\textbullet};
\draw [->, dotted] (3b) to (3a);
\node [above] (3t3b) at (3b.north) {$v_k$};
\node (3c) at (6, -1) {$\textbf{v}$};
\node [below] (3p1) at (3c.south) {$p, m$};

\draw [->] (2a) to (2);
\draw [->] (2a) to (2);
\draw [->] (2) to (1);
\draw [->] (3) to (1);
\node [draw, dotted, thick, fit=(2) (2a) (2b) (2t1) (2t2a) (2t2b) ] (u) {};
\node [draw ,dotted, thick,fit=(3) (3a) (3b) (3t1) (3t3a) (3t3b)] (v) {};
\draw [->,decorate,decoration={snake,amplitude=.4mm,segment
  length=2mm, post length=1mm}] (v) to (u);
\end{configuration}
\end{tabular}
\end{center}


In the Lewis model, $u_1$ is the world most similar to $w$, and in $u_1$ the missile successfully launches, there is a nuclear holocaust, and so the Nixon
counterfactual is true. Lewis thus responds to Fine by defending a
similarity order that favours $u_1$ over $v_1$. He is justified by
prioritising perfect over approximate match in a similarity
relation according to the aforementioned system. 

The interpretation of the Nixon counterfactual we will offer is in line with Lewis', though we do not rely on his system of priorities. We will achieve a resolution similar to his without having to defend a model different from Fine's. After all, as Lewis says: ``I do not claim that this pre-eminence of perfect match is intuitively obvious. I do not claim that it is a feature of the similarity relations most likely to guide our explicit judgments. It is not; else the objection we are considering never would have been put forward.''\cite[p.~470]{lewis1979counterfactual} Instead, we will treat the Nixon counterfactual with an explicit ceteris paribus clause, dispatching with the unintuitive pre-eminence of perfect match in constructing the similarity relation.  

Our interpretation of the Nixon counterfactual is much like in preference logic, where formal ceteris paribus reasoning was first applied \cite{RefWorks:373,
  doyle1994representing, RefWorks:2}. Consider the following diagram, which
shows a preference of a raincoat to an umbrella, provided wearing
boots is kept constant: 
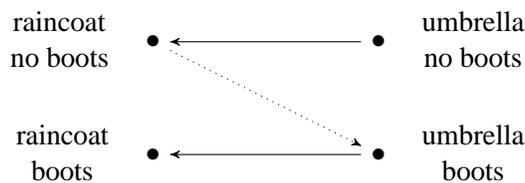
\begin{center}
\begin{tabular}{c}
\begin{configuration}
\node (1) at (0,0) {\textbullet};
\node (1a) [left] at (1.west) {\begin{tabular}{c}raincoat\\no boots\end{tabular}};
\node (2) at (3,0) {\textbullet};
\node (2a) [right] at (2.east) {\begin{tabular}{c}umbrella\\no boots\end{tabular}};
\node (3) at (0, -1.5) {\textbullet};
\node (3a) [left] at (3.west) {\begin{tabular}{c}raincoat\\boots\end{tabular}};
\node (4) at (3, -1.5) {\textbullet};
\node (4a) [right] at (4.east) {\begin{tabular}{c}umbrella\\boots\end{tabular}};
\draw [->] (2) to (1);
\draw [->] (4) to (3);
\draw [->, dotted] (1) to (4);
\end{configuration}
\end{tabular}
\end{center}
Arrows point to more preferred alternatives, and are transitive. Evidently, having an umbrella and boots is preferred to having a raincoat and no
boots. The variation of having boots \emph{skews} the preference. If a
ceteris paribus clause is enforced, guaranteeing that in either case
boots will be worn or boots will not be worn, then the correct
preference is recovered. A similar situation occurs in the logic of
counterfactuals. The variation of certain propositions can skew the
similarity order. In Fine's argument, this is done by the variation of
physical law, a miracle. If we were to restrict the worlds considered
during the evaluation of the counterfactual to those that agree with
$w$ on the proposition $m$, then in $\mathcal{F}$ the world $v_1$
would no longer assume the role of minimal $p$-world. Rather, $u_1$
would. In world $u_1$ a nuclear holocaust \emph{does} occur, whence
the counterfactual becomes true, as desired. This is our resolution of
the Nixon argument, which we next formalise.

\section{Ceteris paribus semantics}
We introduce into our language a new conditional operator which
  generalises the usual one. In particular, it accommodates explicit
  ceteris paribus clauses. The authors in \cite{RefWorks:2} were the first to
  define object languages in this way. They developed a modal logic of
  ceteris paribus preferences in the sense of von Wright
  \cite{RefWorks:373}. For now we will take the ordinary
  conditional operator and embed within it a finite set of formulas
  $\Gamma$ understood as containing the \emph{other things} to be kept
  equal.\footnote{The choice of $\Gamma$ finite is largely technical. We will mention some possibilities and difficulties regarding the case where the ceteris paribus set $\Gamma$ may be infinite in our concluding remarks.} 

\begin{definition}[Language $\mathcal{L}_\CP$] \label{cpcflanguage}
Let $\Gamma$ be a finite set of formulas. Then the language $\mathcal{L}_\CP$ is given by the grammar\footnote{We redefine the language more precisely as Definition \ref{fulldefn} in the appendix. For simplicity we work with the one now stated.} 
\begin{align*} \phi \ ::= \ p \ | \ \lnot \phi \ | \ \phi \lor \psi \ | \ [\phi, \Gamma]\psi.\end{align*}
\end{definition}
We understand the modality $[\phi, \Gamma]\psi$ as the counterfactual
$\phi \boxright \psi$ subject to the requirement that the truth of the
formulas in $\Gamma$ does not change. We define $\phi \land \psi :=
\lnot(\lnot \phi \lor \lnot \psi)$, $\phi \rightarrow \psi := \lnot
\phi \lor \psi$, $\langle \phi, \Gamma \rangle \psi := \lnot [\phi, \Gamma] \lnot \psi$. We call the conditional $[\phi, \Gamma]\psi$ a \emph{ceteris paribus conditional}, or, if the antecedent is false, a \emph{ceteris paribus counterfactual}. $\mathcal{L}_\CP$ is interpreted over standard conditional models, and thus requires no additional semantic information.

Some additional notation is required, however. Let $\M = (W, \preceq, V)$ be a conditional model and let $w, u, v \in W$. Let $\Gamma \subseteq\mathcal{L}_\CP$ be finite. 
\begin{itemize}
\item Define the relation $\equiv_\Gamma$ over $W$ by $u \equiv_\Gamma v$ if for all $\gamma \in \Gamma$, $\M, u \models \gamma$ iff $\M, v \models \gamma$. Then $\equiv_\Gamma$ is an equivalence relation.\footnote{Technically, the relation $\equiv_\Gamma$ should be defined together with the semantics in Definition \ref{defcpcf} by mutual recursion. Again, we favour the simpler presentation. }
\item Set $[w]_\Gamma = \{u \in W_w : w \equiv_\Gamma u\}$, the collection of $w$-entertainable worlds which agree with $w$ on $\Gamma$.
\item Define $\unlhd^\Gamma_w \, := \, \preceq_w \cap \ ([w]_\Gamma \times [w]_\Gamma)$, the restriction of $\preceq_w$ to the above worlds.
\end{itemize}
Thus if $u, v \in [w]_\Gamma$ then either $u \unlhd^\Gamma_w v$ or $v \unlhd^\Gamma_w u$.

\begin{definition}[Semantics] \label{defcpcf}
Let $\M = (W, \preceq, V)$ be a conditional model. Then
\begin{center}
\item\begin{tabular}{r c l}
$\lb [\phi, \Gamma]\psi \rb^\M$ & = & $\{w \in W : \min_{\unlhd^\Gamma_w}(\lb \phi \rb^\M) \subseteq \lb \psi \rb^\M\}$.
\end{tabular} 
\end{center}
\end{definition}
The semantics for the regular connectives are the same as those in
Definition \ref{defcfwf}. Notice that we recover the ordinary counterfactual $\phi \boxright \psi$ with
 $[\phi, \emptyset]\psi$.

Consider again the Fine model $\F$. As before we have $\F, w \not \models p \boxright h$, but now \begin{align} \label{eq1} \F, w \models [p, \{m\}]h. \end{align}

We thus think about the Nixon counterfactual by way of ceteris paribus reasoning. Allowing the truth of arbitrary formulas to vary during the
evaluation of a counterfactual can distort the given similarity order,
thereby attributing falsity to a sentence which may be intuitively true. By forcing certain formulas to keep their truth status fixed one can rule out these cases, which has just been demonstrated with (\ref{eq1}). This ceteris paribus qualification is done in preference logic, and indeed in more general scientific and economic practice.\footnote{See Schurz \cite{RefWorks:401} on comparative ceteris paribus laws.} The Nixon counterfactual is simply a situation involving a defeater, or an irrelevant alternative, which ought to be forced out.

\section{Ceteris paribus as a dynamic action} \label{dynamicop}

The modality $[\phi, \Gamma]\psi$ behaves like a dynamic operator, in
the sense of dynamic epistemic logic. For modality-free formulas
$\phi$ and $\psi$, evaluating $[\phi, \Gamma]\psi$ at $w \in W$ amounts to transforming $$\M = (W,\{\preceq_w\}_{w \in W}, V)$$ into $$[\Gamma]\M = (W, \{\unlhd^\Gamma_w\}_{w \in W}, V)$$ and evaluating $\phi \boxright \psi$ at $[\Gamma]\M, w$. This dynamic action is possible since we are altering the relational structure of $\M$ with only a finite amount of information from $\Gamma$.

\begin{figure}[h]
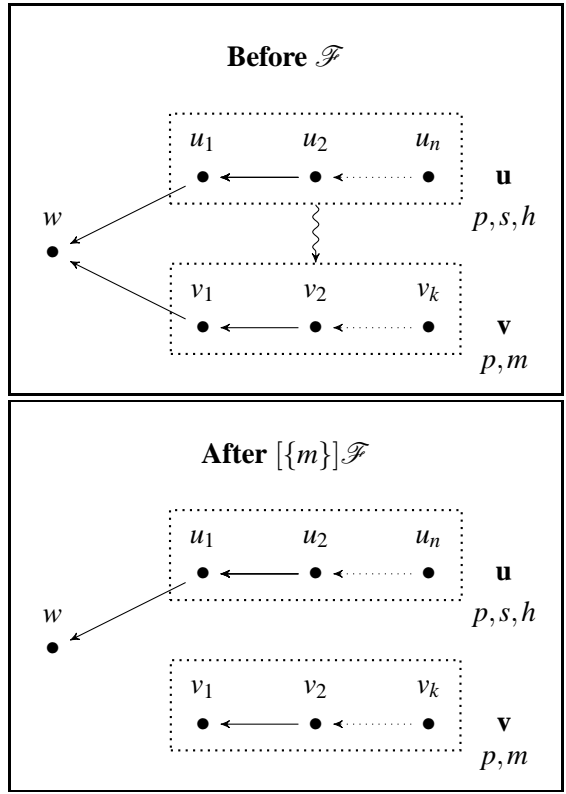

\begin{center}
\begin{tabular}{|c|}
\hline \\
\textbf{Before} $\mathcal{F}$
\\
\
\\
\begin{configuration}
\node (1) at (0,0) {\textbullet};
\node [above] at (1.north) {$w$};
\node (2) at (2,1) {\textbullet};
\node [above] (2t1) at (2.north) {$u_1$};
\node (2a) at (3.5, 1) {\textbullet};
\draw [->] (2a) to (2);
\node [above] (2t2a) at (2a.north) {$u_2$};
\node (2b) at (5, 1) {\textbullet};
\draw [->, dotted] (2b) to (2a);
\node [above] (2t2b) at (2b.north) {$u_n$};
\node (2c) at (6, 1) {$\textbf{u}$};
\node [below] (2p1) at (2c.south) {$p, s, h$};
\node (3) at (2, -1) {\textbullet};
\node [above] (3t1) at (3.north) {$v_1$};
\node (3a) at (3.5, -1) {\textbullet};
\draw [->] (3a) to (3);
\node [above] (3t3a) at (3a.north) {$v_2$};
\node (3b) at (5, -1) {\textbullet};
\draw [->, dotted] (3b) to (3a);
\node [above] (3t3b) at (3b.north) {$v_k$};
\node (3c) at (6, -1) {$\textbf{v}$};
\node [below] (3p1) at (3c.south) {$p, m$};
\draw [->] (2a) to (2);
\draw [->] (2a) to (2);
\draw [->] (2) to (1);
\draw [->] (3) to (1);
\node [draw, dotted, thick, fit=(2) (2a) (2b) (2t1) (2t2a) (2t2b) ] (u) {};
\node [draw ,dotted, thick,fit=(3) (3a) (3b) (3t1) (3t3a) (3t3b)] (v) {};
\draw [->,decorate,decoration={snake,amplitude=.4mm,segment
  length=2mm, post length=1mm}] (u) to (v);
\end{configuration}
\\
\hline \hline\\
\textbf{After} $[\{m\}]\mathcal{F}$
\\
\
\\
\begin{configuration}
\node (1) at (0,0) {\textbullet};
\node [above] at (1.north) {$w$};
\node (2) at (2,1) {\textbullet};
\node [above] (2t1) at (2.north) {$u_1$};
\node (2a) at (3.5, 1) {\textbullet};
\draw [->] (2a) to (2);
\node [above] (2t2a) at (2a.north) {$u_2$};
\node (2b) at (5, 1) {\textbullet};
\draw [->, dotted] (2b) to (2a);
\node [above] (2t2b) at (2b.north) {$u_n$};
\node (2c) at (6, 1) {$\textbf{u}$};
\node [below] (2p1) at (2c.south) {$p, s, h$};
\node (3) at (2, -1) {\textbullet};
\node [above] (3t1) at (3.north) {$v_1$};
\node (3a) at (3.5, -1) {\textbullet};
\draw [->] (3a) to (3);
\node [above] (3t3a) at (3a.north) {$v_2$};
\node (3b) at (5, -1) {\textbullet};
\draw [->, dotted] (3b) to (3a);
\node [above] (3t3b) at (3b.north) {$v_k$};
\node (3c) at (6, -1) {$\textbf{v}$};
\node [below] (3p1) at (3c.south) {$p, m$};
\draw [->] (2a) to (2);
\draw [->] (2a) to (2);
\draw [->] (2) to (1);
\node [draw, dotted, thick, fit=(2) (2a) (2b) (2t1) (2t2a) (2t2b) ] (u) {};
\node [draw ,dotted, thick,fit=(3) (3a) (3b) (3t1) (3t3a) (3t3b)] (v) {};
\end{configuration}\\
\hline
\end{tabular}
\end{center}
\caption{The Fine model before and after $\preceq_{w}$ is upgraded to $\unlhd^{\{m\}}_{w}$.}
\label{fig2}
\end{figure}

Note that the set $W_w$ on which $\preceq_w$ is defined on may change after the update. By updating the model $\M$ with a ceteris paribus clause $\Gamma$, worlds which disagree on $\Gamma$ are relegated to the class $W \setminus W_w$ of infinitely dissimilar (indeed, \emph{irrelevant}) worlds. Figure \ref{fig2} shows how the Fine model changes after being updated by a ceteris paribus clause forcing agreement on $m$. This forces out the $\textbf{v}$-worlds from consideration during the evaluation of the counterfactual; in some sense syntactically `correcting' the provided similarity order. Of course, if each world already agreed with $w$ on $\{m\}$ the ceteris paribus clause would have no effect.

The modality-free condition on $\phi$ and $\psi$ cannot be removed. In particular, one cannot iterate the dynamic ceteris paribus action and retain agreement with the static ceteris paribus counterfactual operator. To see this, consider the example in Figure \ref{noiteration}. Taking $\Gamma = \{s\}$ and $\Delta = \emptyset$, one has $\M, w \models [p, \Gamma][q, \Delta]r$, but $[\Gamma]\M, w \not \models p \boxright [q, \Delta] r$.

\begin{figure}
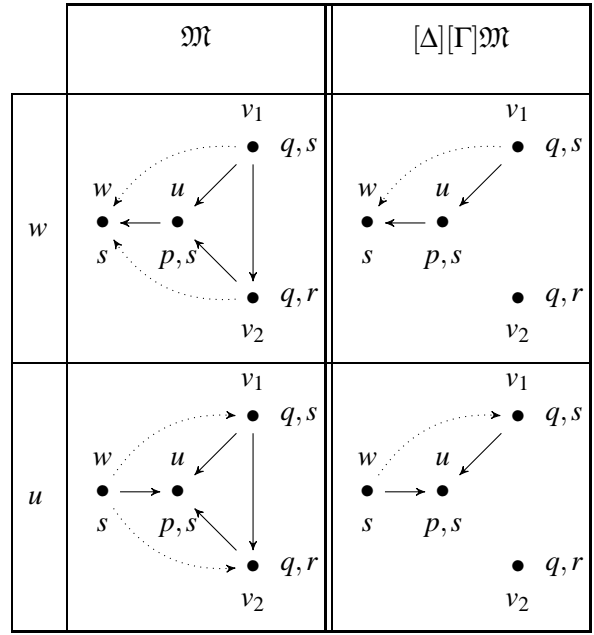
 
\begin{center}
\begin{tabular}{ | m{0.3cm} | m{3cm} | | m{3cm} |}
\cline{2-3}
\multicolumn{0}{r |}{} & \begin{center} $\M$ \end{center} & \begin{center} $[\Delta][\Gamma]\M$ \end{center} \\
\hline
$w$ &
\begin{configuration}
\node (1) at (0,0) {\textbullet};
\node [above] at (1.north) {$w$};
\node [below] at (1.south) {$s$};
\node (2) at (1,0) {\textbullet};
\draw [->] (2) to (1);
\node [above] (2t1) at (2.north) {$u$};
\node [below] at (2.south) {$p,s$};
\node (3) at (2, 1) {\textbullet};
\draw [->] (3) to (2);
\draw [->, dotted, bend right] (3) to (1);
\node [above] (2t2a) at (3.north) {$v_1$};
\node [right] at (3.east) {$q,s$};
\node (4) at (2, -1) {\textbullet};
\draw [<-] (4) to (3);
\draw [->] (4) to (2);
\node [below] (2t2b) at (4.south) {$v_2$};
\draw [->, dotted, bend left] (4) to (1);
\node [right] at (4.east) {$q,r$};
\end{configuration} 
& 
\begin{configuration}
\node (1) at (0,0) {\textbullet};
\node [above] at (1.north) {$w$};
\node [below] at (1.south) {$s$};
\node (2) at (1,0) {\textbullet};
\node [above] (2t1) at (2.north) {$u$};
\node [below] at (2.south) {$p, s$};
\node (3) at (2, 1) {\textbullet};
\draw [->] (2) to (1);
\draw [->] (3) to (2);
\node [above] (2t2a) at (3.north) {$v_1$};
\node [right] at (3.east) {$q, s$};
\node (4) at (2, -1) {\textbullet};
\node [below] (2t2b) at (4.south) {$v_2$};
\node [right] at (4.east) {$q, r$};
\draw [->, dotted, bend right] (3) to (1);
\end{configuration}
\\
\hline
$u$ & 
\begin{configuration}
\node (1) at (0,0) {\textbullet};
\node [above] at (1.north) {$w$};
\node (2) at (1,0) {\textbullet};
\node [below] at (1.south) {$s$};
\node [above] (2t1) at (2.north) {$u$};
\node (3) at (2, 1) {\textbullet};
\node [below] at (2.south) {$p,s$};
\draw [->] (1) to (2);
\draw [->] (3) to (2);
\node [above] (2t2a) at (3.north) {$v_1$};
\node (4) at (2, -1) {\textbullet};
\draw [->] (3) to (4);
\node [right] at (3.east) {$q,s$};
\node [below] (2t2b) at (4.south) {$v_2$};
\draw [->] (4) to (2);
\draw [->, dotted, bend right] (1) to (4);
\draw [->, dotted, bend left] (1) to (3);
\node [right] at (4.east) {$q,r$};
\end{configuration}
&
\begin{configuration}
\node (1) at (0,0) {\textbullet};
\node [above] at (1.north) {$w$};
\node [below] at (1.south) {$s$};
\node (2) at (1,0) {\textbullet};
\node [above] (2t1) at (2.north) {$u$};
\node [below] at (2.south) {$p,s$};
\node (3) at (2, 1) {\textbullet};
\draw [->] (1) to (2);
\draw [->] (3) to (2);
\node [above] (2t2a) at (3.north) {$v_1$};
\node [right] at (3.east) {$q, s$};
\node (4) at (2, -1) {\textbullet};
\node [below] (2t2b) at (4.south) {$v_2$};
\node [right] at (4.east) {$q, r$};
\draw [->, dotted, bend left] (1) to (3);
\end{configuration}
\\
\hline
\end{tabular}
\end{center}
\caption{The horizontal panels labelled $w$ and $u$ define the similarity orders $\preceq_w$ and $\preceq_u$ respectively.}
\label{noiteration}
\end{figure}

\section{Uniformly selecting ceteris paribus clauses}


Having created a formalism which accommodates explicit ceteris paribus clauses, one
might desire a method for uniformly selecting the ceteris paribus set
$\Gamma$. For von Wright \cite{RefWorks:373}, ceteris paribus means
fixing every propositional variable which does not occur in the universe of discourse of the ceteris paribus expression under consideration. More precisely, let $\ud(\phi)$ be the set of all propositional variables occurring in the formula $\phi$, defined inductively as follows.

\begin{center}
\begin{tabular}{l c l}
$\ud(p)$ & = & $\{p\}$ \\
$\ud(\lnot \phi)$ & = & $\ud(\phi)$ \\
$\ud(\phi \lor \psi)$ & = & $\ud(\phi) \cup \ud(\psi)$ \\
$\ud([\phi, \Gamma]\psi)$ & = & $\ud(\phi) \cup \ud(\Gamma) \cup \ud(\psi)$ \\
$\ud(\{\gamma_1, \dots, \gamma_n\})$ & = & $\ud(\gamma_1) \cup \dots \cup \ud(\gamma_n)$.
\end{tabular}
\end{center}
\
\\

Then the ceteris paribus counterfactual \emph{if $\phi$ were the case
  then, ceteris paribus, $\psi$ would be the case} amounts to the
expression \begin{align}\label{vwcp}[\phi, \prop \setminus (\ud(\phi)
             \cup \ud(\psi))]\psi.\end{align} Now all propositional
           variables not occurring in the universe of discourse of the counterfactual antecedent or consequent are fixed. 

Updating the Fine model with respect to von Wright's ceteris paribus
set yields the following model:

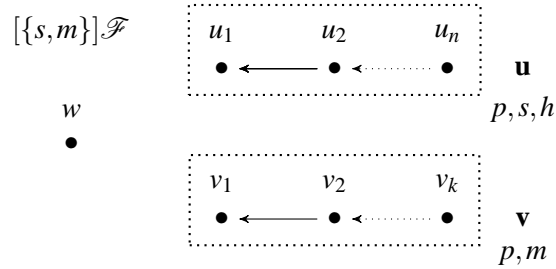
\begin{center}
\begin{tabular}{c}
\begin{configuration}
\node (1) at (0,0) {\textbullet};
\node [above] at (1.north) {$w$};
\node (M) at (0,1.5) {$[\{s, m\}]\mathcal{F}$};
\node (2) at (2,1) {\textbullet};
\node [above] (2t1) at (2.north) {$u_1$};
\node (2a) at (3.5, 1) {\textbullet};
\draw [->] (2a) to (2);
\node [above] (2t2a) at (2a.north) {$u_2$};
\node (2b) at (5, 1) {\textbullet};
\draw [->, dotted] (2b) to (2a);
\node [above] (2t2b) at (2b.north) {$u_n$};
\node (2c) at (6, 1) {$\textbf{u}$};
\node [below] (2p1) at (2c.south) {$p, s, h$};
\node (3) at (2, -1) {\textbullet};
\node [above] (3t1) at (3.north) {$v_1$};
\node (3a) at (3.5, -1) {\textbullet};
\draw [->] (3a) to (3);
\node [above] (3t3a) at (3a.north) {$v_2$};
\node (3b) at (5, -1) {\textbullet};
\draw [->, dotted] (3b) to (3a);
\node [above] (3t3b) at (3b.north) {$v_k$};
\node (3c) at (6, -1) {$\textbf{v}$};
\node [below] (3p1) at (3c.south) {$p, m$};
\draw [->] (2a) to (2);
\draw [->] (2a) to (2);
\node [draw, dotted, thick, fit=(2) (2a) (2b) (2t1) (2t2a) (2t2b) ] (u) {};
\node [draw ,dotted, thick,fit=(3) (3a) (3b) (3t1) (3t3a) (3t3b)] (v) {};
\end{configuration}
\end{tabular}
\end{center}

We have $\F, w \models [p, \{m,s\}]h$, but vacuously! It appears that the relation $\unlhd^\Gamma$ is too strong to interact with von Wright's definition. We are requiring that \emph{everything else} is kept equal. This is questionable metaphysics, to say the least. Lewis made a similar observation in \cite{RefWorks:170}, about the counterfactual  \emph{`if kangaroos had no tails, they would topple over'}: 

\begin{quote} We might think it best to confine our attention to
  worlds where kangaroos have no tails and \emph{everything} else is
  as it actually is; but there are no such worlds. Are we to suppose
  that kangaroos have no tails but that their tracks in the sand are
  as they actually are? Then we shall have to suppose that these
  tracks are produced in a way quite different from the actual
  way. [...] Are we to suppose that kangaroos have no tails but that their genetic makeup is as it actually is? Then we shall have to suppose that genes control growth in a way quite different from the actual way (or else that there is something, unlike anything there actually is, that removes the tails). And so it goes; respects of similarity and difference trade off. If we try too hard for exact similarity to the actual world in one respect, we will get excessive differences in some other respect. (\cite[p.~9]{RefWorks:170})
\end{quote}

In fact, for the logic of \emph{ceteris paribus} counterfactuals to function in
a meaningful fashion, every formula occurring in $\Gamma$
must be independent from the counterfactual antecedent. In the Fine model, we insist that the truth values of $s$ and $m$ are
kept fixed. These propositions, however, are nomologically related to $p$, so we
can't change the truth value of $p$ without affecting the truth values of $s$ and $m$. This is 
why the counterfactual $[p, \{m,s\}]h$ is vacuously true, but
then so is the counterfactual $[p, \{m,s\}]\lnot h$.  To accommodate a uniform method for selecting ceteris paribus clauses, more flexibility is required.  What \emph{ought} to be kept equal when we can't
keep \emph{everything else} equal? In the next section we will consider two strategies for relaxing the interpretation of \emph{ceteris paribus} to address this question.



\section{Relaxing the ceteris paribus clause}

\subsection{Na\"{i}ve counting}

We will now introduce another interpretation for the modality $[\phi, \Gamma]\psi$. Let us write $\lb [\phi, \Gamma]\psi \rb^\M_{\AT}$ for the set $\lb [\phi, \Gamma]\psi \rb^\M$ from Definition \ref{defcpcf}, and let $\models_{\AT}$ act as the ordinary satisfaction relation for Boolean formulas, but with \begin{center}$\M, w \models_{\AT} [\phi, \Gamma]\psi$ iff $w \in \lb [\phi, \Gamma]\psi \rb^\M_{\AT}$. \end{center}

Whereas in Definition \ref{defcpcf} we required strict agreement on the set $\Gamma$, in order to develop a logic for ceteris paribus counterfactuals with a weaker semantics we will instead relax the requirement to \emph{maximal agreement}. The best we can do is preserve the set $\Gamma$ as much as possible for any given model.

Let $\Gamma \subseteq \mathcal{L}_\CP$ be finite, and let $\M = (W, \preceq, V)$ be a conditional model. Define $A_\Gamma^\M: W \times W \to 2^\Gamma$ by \begin{align}A_\Gamma^\M(u, v) = \{\gamma \in \Gamma: \M, u \models \gamma \text{ iff } \M, v \models \gamma\}. \end{align}
Define the relation $\preceq_w^\Gamma$ on $W_w$ by $u \preceq_w^\Gamma v$ iff 
\begin{center} $\text{either } |A_\Gamma^\M(u, w)| > |A_\Gamma^\M(v, w)|, \text{ or } |A_\Gamma^\M(u, w)| = |A_\Gamma^\M(v, w)| \text{ and }u \preceq_w v. $
\end{center}

The relation $\preceq_w^\Gamma$ can be seen as a transformed $\preceq_w$, reordering the similarity order so that worlds closer to $w$ preserve at least as much of $\Gamma$ as worlds further away, and if any two worlds agree on $\Gamma$ to the same quantity, then the nearer world is more similar to $w$ with respect to $\preceq$.

\begin{definition}[Semantics] Let $\M = (W, \preceq, V)$ be a conditional model satisfying the limit assumption. Let $\Gamma \subseteq \mathcal{L}_\CP$ be finite. Then \\
\begin{center}
\begin{tabular}{r c l}
$\lb [\phi, \Gamma]\psi \rb^\M_\NC$ & = & $\{w \in W : \min_{\preceq^\Gamma_w}(\lb \phi \rb^\M) \subseteq \lb \psi \rb^\M\}$.
\end{tabular}
\end{center}
\end{definition}

We write $\M, w \models_{\NC} [\phi, \Gamma]\psi$ iff $w \in \lb [\phi, \Gamma]\psi \rb^\M_{\NC}$.

\begin{fact} Let $\M = (W, \preceq, V)$ be a conditional model. Let $w \in W$, and let $\mathbf{X} \in \{\AT, \NC\}$. Then the following are true, where $\pm \alpha$ is shorthand which uniformly stands for either $\alpha$ or $\lnot \alpha$:  \label{cf} \
\begin{enumerate}
\item $\M, w \models \phi \boxright \psi$ iff $\M, w \models_{\mathbf{X}}[\phi, \emptyset]\psi$
\item $\M, w \models_\mathbf{X} (\pm \alpha \land \langle \varphi, \Gamma \rangle (\pm \alpha \land \psi)) \rightarrow \langle \varphi, \Gamma \cup \{\alpha\} \rangle \psi$ 
\item $\M, w \models_{\AT} \langle \phi, \Gamma \rangle \psi  \Rightarrow \M, w \models_{\NC} \langle \phi, \Gamma \rangle \psi$ 
\item $\M, w \models_{\NC} [\phi, \Gamma]\psi  \Rightarrow \M, w \models_{\AT}[\phi, \Gamma]\psi$ 
\end{enumerate}

\end{fact}

The original ceteris paribus preference logic \cite{RefWorks:2} could be axiomatised using standard axioms together with Fact 1.2 and its converse. A crucial difference with \NC \ semantics is that the converse of Fact 1.2 does not hold. The existence of a $\phi \land\psi$-world which maximally agrees on $\Gamma \cup \{\alpha\}$ does not ensure that $\alpha$ actually holds at that world. In fact, it is not guaranteed that \emph{any} formula from $\Gamma \cup \{\alpha\}$ is obtained. 

\subsection{Maximal supersets}

An approach to counterfactuals familiar to the AI community
\cite{katsuno1991propositional, dalal1988investigations,
  del1994interference, del1996belief} makes use of a selection
function which chooses the `closest' world according to maximal sets
of propositional variables. More specifically, each world $w$ satisfies some set $\textbf{P}_w \subseteq \prop$ of propositional variables, and a world $u$ is a world closest to $w$ if there is no $v$ with $\textbf{P}_u \subset \textbf{P}_v \subseteq \textbf{P}_w$. Taking this as a kind of ceteris paribus formalism we obtain the following variant of our ceteris paribus counterfactuals. First let us define the relation $\sqsubseteq_w^\Gamma$ on $W_w$ by $u \sqsubseteq_w^\Gamma v$ iff
\begin{center} $\text{either } A_\Gamma^\M(v, w) \subset A_\Gamma^\M(u, w), \text{ or }  A_\Gamma^\M(v, w) = A_\Gamma^\M(u, w) \text{ and }u \preceq_w v. $
\end{center}

\begin{definition}[Semantics] Let $\M = (W, \preceq, V)$ be a conditional model satisfying the limit assumption. Let $\Gamma \subset \mathcal{L}_\CP$ be finite. Then 
\
\\
\begin{center}
\begin{tabular}{r c l}
$\lb [\phi, \Gamma]\psi \rb^\M_\MS$ & = & $\{w \in W : \min_{\sqsubseteq^\Gamma_w}(\lb \phi \rb^\M) \subseteq \lb \psi \rb^\M\}$.
\end{tabular}
\end{center}
\end{definition}
We write $\mathfrak{M}, w \models_\MS [\phi, \Gamma]\psi$ iff $w \in \lb [\phi, \Gamma]\psi \rb^\M_\MS$. Now $\Gamma$ is maximally preserved in the sense that worlds which preserve the same propositions as another, and furthermore preserve additional propositions from $\Gamma$, are deemed to approximate $\Gamma$ more closely; while worlds $u, v$ with neither $A^\mathfrak{M}_\Gamma(u, w) \subseteq A^\mathfrak{M}_\Gamma(v, w)$ nor $A^\mathfrak{M}_\Gamma(v, w) \subseteq A^\mathfrak{M}_\Gamma(u, w)$ are considered incomparable.

\begin{fact} [Extends Fact \ref{cf}] Let $\M = (W, \preceq, V)$ be a conditional model. Let $w \in W$. Then the following are true.
\begin{enumerate}
\item $\M, w \models \phi \boxright \psi$ iff $\M, w \models_{\MS}[\phi, \emptyset]\psi$
\item $\M, w \models_\MS (\pm \alpha \land \langle \varphi, \Gamma \rangle (\pm \alpha \land \psi)) \rightarrow \langle \varphi, \Gamma \cup \{\alpha\} \rangle \psi$ 
\item $\M, w \models_{\AT} \langle \phi, \Gamma \rangle \psi  \Rightarrow \M, w \models_{\MS} \langle \phi, \Gamma \rangle \psi$ 
\item $\M, w \models_{\MS} [\phi, \Gamma]\psi  \Rightarrow \M, w \models_{\AT}[\phi, \Gamma]\psi$ 
\end{enumerate}
\end{fact}

\section{Dynamics and the Nixon counterfactual}

Given a ceteris paribus interpretation $\mathbf{X} \in \{\AT, \NC,
\MS\}$, let us write $[\Gamma]_\mathbf{X}\mathfrak{M}$ for the model
$\mathfrak{M}$ updated with a ceteris paribus clause $\Gamma$
according to interpretation $\mathbf{X}$. Specifically, we have the following definition.

\begin{definition}
Let $\mathfrak{M} = (W, \preceq, V)$ be a conditional model, and let $\Gamma \subseteq \mathcal{L}_\CP$ be a finite set of formulas. We define the \emph{updated models} $[\Gamma]_\mathbf{X}\mathfrak{M}$, for $\mathbf{X} \in \{\AT, \NC, \MS\}$, by \\
\begin{center}
\begin{tabular} {l c l}
$[\Gamma]_\AT\mathfrak{M}$ & $:=$  & $(W, \unlhd^\Gamma, V);$ \\
$[\Gamma]_\NC\mathfrak{M}$ & $:=$ & $(W, \preceq^\Gamma, V);$ \\
$[\Gamma]_\MS\mathfrak{M}$ & $:=$ & $(W, \sqsubseteq^\Gamma, V).$
\end{tabular}
\end{center}
\end{definition}

This provides us with three dynamic \emph{ceteris paribus} updates. Let us see how they treat the Nixon counterfactual. We have already witnessed the $\AT$ update with ceteris paribus sets $\{m\}$ and $\{m,s\}$, and concluded that both make the counterfactual true (vacuous truth with $\{m, s\}$). $\NC$ and $\MS$ updates agree on the truth of the Nixon counterfactual with the $\AT$ update on $\{m\}$, but disagree on $\{m,s\}$. Updating the Fine model with von Wright's ceteris paribus clause $\{m, s\}$ according to the $\NC$ interpretation yields $\F$ again. Thus $\F, w \not\models_\NC [p, \{m,s\}]h$. Updating Fine's model with $\{m, s\}$ according to the $\MS$ interpretation gives the following model:

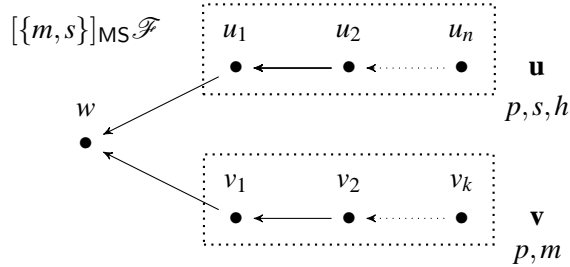
\begin{center}
\begin{tabular}{c}
\begin{configuration}
\node (1) at (0,0) {\textbullet};
\node [above] at (1.north) {$w$};

\node (M) at (0,1.5) {$[\{m,s\}]_\MS\mathcal{F}$};

\node (2) at (2,1) {\textbullet};
\node [above] (2t1) at (2.north) {$u_1$};
\node (2a) at (3.5, 1) {\textbullet};
\draw [->] (2a) to (2);
\node [above] (2t2a) at (2a.north) {$u_2$};
\node (2b) at (5, 1) {\textbullet};
\draw [->, dotted] (2b) to (2a);
\node [above] (2t2b) at (2b.north) {$u_n$};
\node (2c) at (6, 1) {$\textbf{u}$};
\node [below] (2p1) at (2c.south) {$p, s, h$};

\node (3) at (2, -1) {\textbullet};
\node [above] (3t1) at (3.north) {$v_1$};
\node (3a) at (3.5, -1) {\textbullet};
\draw [->] (3a) to (3);
\node [above] (3t3a) at (3a.north) {$v_2$};
\node (3b) at (5, -1) {\textbullet};
\draw [->, dotted] (3b) to (3a);
\node [above] (3t3b) at (3b.north) {$v_k$};
\node (3c) at (6, -1) {$\textbf{v}$};
\node [below] (3p1) at (3c.south) {$p, m$};

\draw [->] (2a) to (2);
\draw [->] (2a) to (2);
\draw [->] (2) to (1);
\draw [->] (3) to (1);
\node [draw, dotted, thick, fit=(2) (2a) (2b) (2t1) (2t2a) (2t2b) ] (u) {};
\node [draw ,dotted, thick, fit=(3) (3a) (3b) (3t1) (3t3a) (3t3b)] (v) {};
\end{configuration}
\end{tabular}
\end{center}

In $[\{m,s\}]_\MS\F$ the Nixon counterfactual is not true, and neither is $p \boxright \lnot h$. 

We summarise the truth of the Nixon counterfactuals $p \boxright h$ and $p \boxright \lnot h$ in the various updated Fine models in the following table. 

\begin{center}
\begin{tabular}{cc|c|c|c|}
\cline{3-5}
                                                                   &        & \multicolumn{3}{c|}{\emph{Interpretation}} \\ \hline
\multicolumn{1}{|c|}{\emph{Counterfactual}}                               & \emph{Clause} & $ \quad \AT \quad$       & $\quad\NC\quad$       & $\quad\MS\quad$       \\ \hline
\multicolumn{1}{|c|}{\multirow{2}{*}{$p \boxright h$}}     & $\{m\}$      & $\mathsf{true}$       & $\mathsf{true}$     & $\mathsf{true}$     \\ \cline{2-5} 
\multicolumn{1}{|c|}{}                                             & $\{m,s\}$     & $\mathsf{true}$       & $\mathsf{false}$     & $\mathsf{false}$     \\ \hline
\multicolumn{1}{|c|}{\multirow{2}{*}{$p \boxright \lnot h$}} & $\{m\}$      & $\mathsf{false}$      & $\mathsf{false}$      & $\mathsf{false}$      \\ \cline{2-5} 
\multicolumn{1}{|c|}{}                                             & $\{m,s\}$     & $\mathsf{true}$       & $\mathsf{true}$      & $\mathsf{false}$     \\ \hline
\end{tabular}
\end{center}

The rows labelled with $p \boxright h$ and $p \boxright \lnot h$ indicate the truth value of those counterfactuals in the updated models $[\Gamma]_\mathbf{X}\F$, where $\Gamma$ is given by the cell in the \emph{Clause} column and $\mathbf{X}$ is given by the \emph{Interpretation} column.

Formally, the table illustrates how different truth values for the Nixon counterfactual may be obtained by combining the various interpretations of ceteris paribus ($\AT, \NC, \MS$) with the different ceteris paribus sets (the selected set $\{m\}$ or von Wright's set $\{m, s\}$). But this doesn't mean that all combinations are legitimate formalisations of Fine's argument. Fine's story is about small miracles that can interfere with Nixon's ploy, not about whether the missile would successfully launch should Nixon press the button. That the proposition $s$ must be able to vary is crucial to the story, so one shouldn't attempt to keep it equal, on a par with $m$. 
  We adhere to our favoured formalisation of the Nixon argument in which the proposition $m$ is the only one that needs to be kept equal. We have given principled reasons for this choice, and our selection makes the counterfactual true -- all interpretations agree on that. The point of the table is a formal one, namely that the truth-values of counterfactuals vary with different ceteris paribus updates according to their interpretation.

\section{Theorems}

In the appendix (Corollary \ref{col1}) we prove that the logic $\Lambda^{\mathcal{L}_\CP}_{\mathfrak{C}}$ of ceteris paribus counterfactuals over the class of conditional frames $\mathfrak{C}$ is complete for $\AT/\NC/\MS$ semantics. The proof works by translating formulas of $\mathcal{L}_\CP$ into formulas of a \emph{comparative possibility} language, in the style of Lewis, and axiomatising the equivalent logic. This permits a clearer reduction of ceteris paribus modalities to basic comparative possibility operators, albeit with a translation exponential in the size of $\Gamma$.

\section{Concluding remarks}

This paper has introduced a ceteris paribus logic for counterfactual
reasoning by adapting the formalism in \cite{RefWorks:2}. We have
introduced some variants on ceteris paribus logic in light of
philosophical difficulties arising in the application of
conditionals. We apply our framework to the \emph{Nixon counterfactual}, and with this bring a new perspective to the problem.  We have suggested and explored the dynamic
perspective of our various syntactic interpretations of \emph{ceteris paribus}, which has resulted in a richer understanding of
so-called \emph{comparative} ceteris paribus reasoning in formal settings. We
have provided completeness theorems which demonstrate that the ceteris
paribus logics so obtained ultimately reduce to the underlying
counterfactual logic; in our case Lewis' VC. With our framework we
defend Lewisian semantics by appealing to examples from preference
logic, where ceteris paribus reasoning is more widely discussed.

Finally, we outline some limitations of our framework and directions for future research.

\noindent\emph{Iterated ceteris paribus actions.} We saw in Section \ref{dynamicop} that iterated ceteris paribus counterfactuals deviate in truth-value from the corresponding update-then-counterfactual sequence.  Such difficulties with iterated counterfactuals are not so uncommon. We leave the task of understanding the full interaction between the two for further investigation.

\noindent\emph{Cardinality restrictions on $\Gamma$.} In general, ceteris paribus
reasoning requires keeping equal as much information as possible, and
sometimes unknown information (for example, unanticipated defeaters of
laws). Keeping everything else equal may indeed mean keeping equal an
indefinite, and possibly infinite, set of things. Exploring ceteris
paribus logic without cardinality restrictions to $\Gamma$ is thus
more than a mere technical exercise. But it is not so straightforward
to extend the present framework to accommodate the presence of
infinite $\Gamma$. The translations presented in the appendix
  only carry over to the infinite case for infinitary languages,
  which is not much of a solution. For the strict
  ceteris paribus semantics, we instead suggest following the $\delta$-flexibility
  approach of \cite{RefWorks:390}. For the relaxed ceteris paribus semantics, there
  are conceptual difficulties which arise with the comparison of infinite
  sets: when should we say of two infinite sets that one keeps more
  things equal than the other? Clearly na\"{i}ve counting will not suffice. Minimising distance with respect to $\sqsubseteq^\Gamma$ is more promising, but has its own problems. We leave this challenging technical enterprise for future research.

\section{Acknowledgments} 

We  wish to thank the participants at the Australasian Association of Logic and the Analysis, Randomness and Applications meetings held in New Zealand in 2014. We also wish to thank Sam Baron, Andrew Withy, and the anonymous referees for valuable comments.

\begin{appendix}
\section{Appendix}
We first recast Definition \ref{cpcflanguage} in a more formally precise manner.

\begin{definition} \label{fulldefn} For each ordinal $\alpha$ let $\mathcal{L}_\alpha$ be given by $$\phi \ :: =  \ p \ | \ \bot \ | \ \lnot \phi \ | \ \phi \lor \psi \ | \ [\phi, \Gamma] \psi$$ where $\Gamma \subseteq \mathcal{L}_\beta$ is finite and $\beta < \alpha$. $\mathcal{L}_{\CP}$ is then defined to be $\bigcup_\alpha \mathcal{L}_\alpha$. \end{definition} 

This ensures the sets $\Gamma$ are well-defined. One can define a language $\mathcal{L}$ of comparative possibility in a similar style, though we will only give the following grammar
\begin{center} $\phi \ ::=  \ p \ | \ \bot \ | \ \lnot \phi \ | \ \phi \lor \psi \ | \ \phi \preceq \psi \ | \ \phi \preceq^\Gamma \psi \ | \ \phi \unlhd^\Gamma \psi | \ \phi \sqsubseteq^\Gamma \psi. $\end{center}

\noindent We further set

\begin{tabular}{r c l r c l r c l}
$\phi \prec \psi$ & $:=$ & $\lnot (\psi \preceq \phi)$; & $\phi \prec^\Gamma \psi$ & $:=$ & $\lnot(\psi \preceq^\Gamma \phi)$; & $\phi \lhd^\Gamma \psi$ & $:=$ & $\lnot(\psi \unlhd^\Gamma \phi);$\\
$\phi \sqsubset^\Gamma \psi$ & $:=$ & $\lnot(\psi \sqsubseteq^\Gamma \phi)$; & $\Diamond \phi$ & $:=$ & $\phi \prec \bot$; & $\square \phi$ & $ := $ & $ \lnot \Diamond \lnot \phi$.\\
\end{tabular}

\begin{definition}[Semantics]\label{def:sem-sim}
Let $\mathfrak{M}, w$ be a conditional model. Then \end{definition}
\begin{tabular}{r c l} 
$\lb p \rb^\M$ & = & $V(p)$;\\
$\lb \bot \rb^\M$ & = & $\emptyset$;\\
$\lb \lnot \phi\rb^\M $ & = & $W\setminus \lb \phi \rb^\M$;\\
$\lb \phi \lor \psi \rb^\M$ & = & $\lb \phi \rb^\M \cup \lb \psi \rb^\M$;\\
$\lb \phi \preceq \psi \rb^\M$ & = & $\{w \in W :\forall u \in W_w \ \exists v \in W_w $ such that if $u \in \lb\psi\rb^\M$ then $v \in \lb\phi\rb^\M$ and $v \preceq_w u\}$; \\
$\lb \phi \preceq^\Gamma \psi \rb^\M$ & = & $\{w \in W :\forall u \in W_w \ \exists v \in W_w $ such that if $u \in \lb\psi\rb^\M$ then $v \in \lb\phi\rb^\M$ and $v \preceq^\Gamma_w u\}$; \\
$\lb \phi \unlhd^\Gamma \psi \rb^\M$ & = & $\{w \in W :\forall u \in [w]_\Gamma \ \exists v \in [w]_\Gamma $ such that if $u \in \lb\psi\rb^\M$ then $v \in \lb\phi\rb^\M$ and $v \unlhd^\Gamma_w u\}$; \\
$\lb \phi \sqsubseteq^\Gamma \psi \rb^\M$ & = & $\{w \in W :\forall u \in W_w \ \exists v \in W_w $ such that if $u \in \lb\psi\rb^\M$ then $v \in \lb\phi\rb^\M$ and $v \sqsubseteq^\Gamma_w u\}$.
\end{tabular}

\begin{lemma} \label{definable}The modal operator $[\phi, \Gamma]\psi$ under $\NC$ semantics is definable in $\mathcal{L}$.
\end{lemma}

\begin{proof}
We show that \begin{align*}
\M, w \models_{\NC} [\phi, \Gamma]\psi \text{ iff } \ \M, w \models \Diamond \phi \rightarrow (\phi \land \psi) \prec^\Gamma (\phi \land \lnot \psi).
\end{align*}
\noindent $\Rightarrow:$ Assume $\M, w \models \Diamond \phi$. Then
there is a world $x \in W_w$ such that $\M, x\models
\phi$. So, by assumption, $\min_{\unlhd^\Gamma_w}(\lb \phi \rb^\M)\, \neq \,
\emptyset$ and $\min_{\unlhd^\Gamma_w}(\lb \phi \rb^\M) \subseteq \lb \psi \rb^\M$. Hence, there exists $y \in W_w$ such that
$\M, y\models \phi \land \psi$ and for every world $z \in W_w$, if
$z\preceq_w^\Gamma y$ then $z \not\in \lb \phi \land \lnot \psi \rb^\M$. This is exactly $\M, w\models \con \phi \psi \prec^\Gamma \con \phi {\lnot\psi}$. 

\noindent $\Leftarrow:$ By contrapositive. Assume $\M, w\not\models [\phi,
\Gamma]\psi$. Then, by the semantic definition, there is an $x \in
\min_{\unlhd^\Gamma_w}(\lb \phi \rb^\M)$ such that $x \not \in \lb
\psi \rb^\M$. So $\M, w\models \Diamond \phi$, and for every $x \in
W_w$, there exists $y\in W_w$ (namely $v$)
such that if $x \in \lb \phi \land \psi \rb^\M$, then $y\preceq_w
^\Gamma x$
and $x\in \lb\phi \land \lnot \psi\rb^\M$. Hence, $\M, w \models (\phi
\land \lnot \psi) \preceq^\Gamma (\phi \land \psi)$, so $\M,
w\not\models (\phi \land \psi) \prec^\Gamma (\phi\land\lnot\psi)$, and
we are done. 
\end{proof}

\begin{lemma} \label{definable2}The modal operator $[\phi, \Gamma]\psi$ under $\AT$ semantics is definable in $\mathcal{L}$.
\end{lemma}

\begin{proof} Replace $\preceq_w^\Gamma$ with $\unlhd_w^\Gamma$ in the above proof to show that the following equivalence
\begin{align*}
\mathfrak{M}, w \models_{\AT} [\phi, \Gamma]\psi \text{ iff } \ \mathfrak{M}, w \models \Diamond \phi \rightarrow (\phi \land \psi) \lhd^\Gamma (\phi \land \lnot \psi).
\end{align*}   
holds. \end{proof}

\begin{lemma} \label{definable3}The modal operator $[\phi, \Gamma]\psi$ under $\MS$ semantics is definable in $\mathcal{L}$.
\end{lemma}

\begin{proof} Replace $\preceq_w^\Gamma$ with $\sqsubseteq_w^\Gamma$ in the above proof to show that the following equivalence holds
\begin{align*}
\mathfrak{M}, w \models_{\MS} [\phi, \Gamma]\psi \text{ iff } \ \mathfrak{M}, w \models \Diamond \phi \rightarrow (\phi \land \psi) \sqsubset^\Gamma (\phi \land \lnot \psi).
\end{align*}    \end{proof}

Denote by $\mathcal{L}^-$ the $\mathcal{L}$-fragment given by
\begin{align*} \phi \ ::= \ p \ | \ \bot \ | \ \lnot \phi \ | \ \phi \lor \psi \ | \ \phi \preceq \psi. \end{align*}

Given a set $\Gamma \subseteq \mathcal{L}$ or $\Gamma \subseteq \mathcal{L}^-$, let $\Gamma^*$ be the set of all possible conjunctions of formulas and negated formulas from $\Gamma$; that is, the set of all $\psi$ such that $\psi = \bigwedge \limits_{\gamma \in \Gamma} \pm \gamma$, where $+\gamma = \gamma$ and $- \gamma = \lnot \gamma$ . So if $\Gamma = \{p, \lnot q\}$ then \begin{align*} \Gamma^*= \{p\land \lnot q, \lnot p \land \lnot q, p \land \lnot \lnot q, \lnot p \land \lnot \lnot q\}. \end{align*}
We will often identity a conjunction $\phi_1 \land \dots \land \phi_n$ with the set $\{\phi_1, \dots, \phi_n\}$.
\begin{lemma} \label{lemma:expressible2}
The modal operator $\unlhd^\Gamma$ of $\mathcal{L}$ is definable in $\mathcal{L}^-$.
\end{lemma}
\begin{proof}
We show that \begin{align}
\phi \unlhd^\Gamma \psi \leftrightarrow \bigwedge \limits_{\gamma \in \Gamma^*} \left[ \gamma \rightarrow (\phi \land \gamma) \preceq (\psi \land \gamma) \right].
\end{align}
$\Rightarrow:$ Without loss of generality write $\M, w \models \gamma$. Let $u \in W_w$ and suppose $\M, u \models \psi \land \gamma$. By hypothesis there exists $v \in [w]_\Gamma$ such that $\M, v \models \phi$ and $v \unlhd_w^\Gamma u$. Now $v \equiv_\Gamma w$, so $\M, v \models \gamma$, and $v \preceq_w u$ as required.

\noindent $\Leftarrow:$ Write $\M, w \models \gamma$. Then $\M, w \models (\phi \land \gamma) \preceq (\psi \land \gamma)$. Let $u \in [w]_\Gamma$ and suppose that $\M, u \models \psi$. Then $\M, u \models \psi \land \gamma$, so there exists $v \in W_w$ with $\M, v \models \phi \land \gamma$ and $v \preceq_w u$. Then $v \equiv_\Gamma w$, and so $v \unlhd_w^\Gamma u$.
  \end{proof}

\begin{lemma} \label{lem5}
The modal operator $\sqsubseteq^\Gamma$ of $\mathcal{L}$ is expressible in $\mathcal{L}^-$.
\end{lemma}

\begin{proof}
We show that 
\begin{align} \label{eq89}
\phi \sqsubseteq^\Gamma \psi \leftrightarrow \bigwedge \limits_{\gamma\in \Gamma^*} \Big( \gamma \rightarrow \bigwedge \limits_{\lambda \subseteq \gamma}\Big[\bigwedge\limits_{\lambda \subset \lambda' \subseteq \gamma} \lnot \Diamond(\phi\land\lambda') \rightarrow (\phi\land\lambda) \preceq (\psi\land\lambda)\Big]\Big)
\end{align}

\noindent $\Rightarrow:$ Suppose $\M, w \models \phi \sqsubseteq^\Gamma \psi$ with $\M, w \models \gamma$, for some $\gamma \in \Gamma^*$. Take $\lambda\subseteq \gamma$ such that 
\begin{equation}\label{eq:a}\M, w \models \bigwedge\limits_{\lambda \subset \lambda' \subseteq \gamma} \lnot \Diamond(\phi\land\lambda').
\end{equation}
Take $v \in W_w$ arbitrary such that $\M, v \models \psi \land
\lambda$. By the hypothesis there is $u \in W_w$ such that $u \sqsubseteq^\Gamma_w v$ and $\M, u \models \phi$. 

Now, $u \sqsubseteq^\Gamma_w v$ implies that 
\begin{center}
$(\dagger)$ \qquad \text{either} $A_\Gamma^\M(v, w) \subset A_\Gamma^\M(u, w),$ \text{or} $A_\Gamma^\M(v, w) = A_\Gamma^\M(u, w)$ and $u \preceq_w v$.
\end{center}


If $A_\Gamma^\M(v,w) \subset A_\Gamma^\M(u,w)$, then $\lambda \subseteq A_\Gamma^\M(v,w)$ implies that $\lambda \subset A_\Gamma^\M(u,w)$. Furthermore, $\M, u\models \phi\land A_\Gamma^\M(u,w)$.  Take $\lambda' := A_\Gamma^\M(u,w)$, then $\M, u \models \phi \land \lambda'$, and hence $$\M, w\models \lnot \bigwedge\limits_{\lambda \subset \lambda' \subseteq \gamma} \lnot \Diamond(\phi\land\lambda'),$$ contradicting (\ref{eq:a}). 


Thus by ($\dagger$), $A_\Gamma^\M(v,w) = A_\Gamma^\M(u,w)$ and $u\preceq_w v$. Finally, since $\lambda \subseteq A_\Gamma^\M(v,w)$ and $\M, u\models \phi \land A_\Gamma^\M(u,w)$, we have that $M, u\models \phi \land \lambda$, as desired. 

\noindent $\Leftarrow:$ Assume the right-hand side of (\ref{eq89}). There is a unique $\gamma \in \Gamma^*$ for which $\M, w \models \gamma$. Take $u\in W_w$ such that $\M, u\models \psi$, and consider $A_\Gamma^\M(u,w)$. 

\paragraph{Case 1} There is an $x\in W_w$ such that $A_\Gamma^\M(u,w) \subset A_\Gamma^\M(x,w)$ and $M, x\models \phi$. Then $x\sqsubseteq_w^\Gamma u$, by definition of $\sqsubseteq_w^\Gamma$.

\paragraph{Case 2} There is no $x\in W_w$ such that $A_\Gamma^\M(u,w) \subset A_\Gamma^\M(x,w)$ and $M, x\models \phi$. Now, if there is $y\in W_w$ and a set of formulas $\lambda'$ with $A_\Gamma^\M(u,w) \subset \lambda' \subseteq \gamma$ such that $\M, y\models \phi \land \lambda'$, then $A_\Gamma^\M(u,w) \subset \lambda' \subseteq A_\Gamma^\M(y,w)$ and $M, y\models \phi$, contradicting our assumption. Hence
$$\M, w\models \bigwedge\limits_{A_\Gamma^\M(u,w) \subset \lambda' \subseteq \gamma} \lnot \Diamond(\phi\land\lambda'),$$
and by taking $\lambda := A^\M_\Gamma(w, u)$, our initial assumption implies that
$$\M, w\models (\phi \land A_\Gamma^\M(u,w))\preceq(\psi \land A_\Gamma^\M(u,w)).$$

Since $\M, u\models \psi \land A_\Gamma^\M(u,w)$, there is an $x\preceq_w u$ such that $\M, x\models \phi \land A_\Gamma^\M(u,w)$. Hence $A_\Gamma^\M(u,w) \subseteq A_\Gamma^\M(x,w)$, and also $A_\Gamma^\M(x,w) \subseteq A_\Gamma^\M(u,w)$ as the containment cannot be proper by the case assumption. So $A_\Gamma^\M(u,w) = A_\Gamma^\M(x,w)$, and since $x\preceq_w u$, one has that $x\sqsubseteq_w^\Gamma u$. 

Hence, both cases imply that there exists an $x\sqsubseteq_w^\Gamma u$ such that $\M, x\models \phi$, as desired.
\end{proof}

\begin{lemma} \label{lem6}
The modal operator $\preceq^\Gamma$ of $\mathcal{L}$ is definable in $\mathcal{L}^-$.
\end{lemma}

\begin{proof}
Replace the subset condition 
\begin{align*}
\lambda \subset \lambda' \subseteq \gamma
\end{align*}
in (\ref{eq89}) with the cardinality condition
\begin{align*}
|\lambda| < |\lambda'| \le |\gamma|
\end{align*}
and repeat the above process.
\end{proof}

Notice that, if $\Gamma \cup \{\phi, \psi\} \subseteq \mathcal{L}^-$, then the right hand sides of the equivalences established above are in $\mathcal{L}^-$. This allows us to apply the translation to a formula from the inside-out, the resulting formula belonging to $\mathcal{L}^-$.

By a \emph{conditional frame} we mean a pair $F = (W, \preceq)$, such that $(F, V)$ is a conditional model for any valuation function $V$. Let $\mathfrak{C}$ be the class of conditional frames. Using the notation from~\cite{RefWorks:130}, we write $\Lambda^\mathfrak{L}_\mathfrak{C}$ for the set of $\mathfrak{L}$-formulas valid over $\mathfrak{C}$.

\begin{theorem}
The logic $\Lambda^{\mathcal{L}}_\mathfrak{C}$ is complete.
\end{theorem}

\begin{proof}
We take as our axiomatisation the axioms for VC \cite{RefWorks:170}, plus the translations from Lemmas \ref{lemma:expressible2}, \ref{lem5}, and \ref{lem6}. 
\end{proof}

\begin{corollary} \label{col1}
The logic $\Lambda^{\mathcal{L}_{\CP}}_\mathfrak{C}$ is complete for $\AT/\NC/\MS$-semantics.
\end{corollary}
\end{appendix}
\end{document}